\documentclass[11pt]{article}

\usepackage{times}
\usepackage{amsmath}
\usepackage{amssymb}
\usepackage{amsthm}
\usepackage{graphicx,epsfig}

\newcommand{\pr}{{\sc Pr}}

\newcommand {\A}{\mathcal A}
\newcommand{\e}{\mathcal E}
\newcommand {\V}{\mathbf v}

\newcommand{\F}{\mathcal{F}^{(2)}}
\newcommand{\M}{\mathcal{M}^{(2)}}

\newcommand{\rv}{\text{{\sc Rev}}}
\newcommand{\B}{\mathbf{b}}
\newcommand{\p}{\mathbf{P}}
\newcommand{\I}{\mathcal{I}}

\newtheorem{thm}{Theorem}

\newtheorem{theorem}{Theorem}[section]

\newtheorem{lemma}[theorem]{Lemma}
\newtheorem{claim}[theorem]{Claim}

\newtheorem{define}{Definition}

\advance \topmargin by -\headheight
\advance \topmargin by -\headsep
\evensidemargin \oddsidemargin
\begin{document}

\title{Constant-Competitive Prior-Free Auction with Ordered Bidders
}

\author{Sayan Bhattacharya\thanks{Department of Computer Science, Duke University, Durham NC 27708. Email: {\tt bsayan@cs.duke.edu}. Supported by NSF awards CCF-0745761 and CCF-1008065.} \and Janardhan Kulkarni\thanks{Department of Computer Science, Duke University, Durham NC 27708. Email: {\tt kulkarni@cs.duke.edu}. Supported by NSF award  IIS-0964560.} \and Xiaoming Xu\thanks{Department of Computer Science, Duke University, Durham NC 27708. Email: {\tt xiaoming@cs.duke.edu}. Supported by NSF award CCF-1008065. }}

\maketitle

\begin{abstract}
A central problem in Microeconomics is to design auctions with good revenue properties. In this setting, the bidders' valuations for the items are private knowledge, but they are drawn from publicly known prior distributions. The goal is to find a truthful auction (no bidder can gain in utility by misreporting her valuation) that maximizes the expected revenue. 

Naturally,  the optimal-auction is sensitive to the prior distributions. An intriguing question is to design a truthful auction that is oblivious to these priors, and yet manages to get a constant factor of the optimal revenue. Such auctions are called {\em prior-free}. 

Goldberg et al. presented a constant-approximate prior-free auction when there are identical copies of an item available in unlimited supply, bidders are unit-demand, and their valuations are drawn from  i.i.d. distributions. The recent work of Leonardi et al. [STOC 2012] generalized this problem to non i.i.d. bidders, assuming that the auctioneer knows the ordering of their reserve prices. Leonardi et al. proposed  a prior-free auction that achieves a $O(\log^* n)$ approximation. We improve upon this result, by giving the first prior-free auction with constant approximation guarantee.  
\end{abstract}

\section{Introduction}

We consider the following problem.  There are multiple bidders  interested in purchasing some items that are being auctioned. Typically, the auction consists of two steps. First, the bidders disclose  their {\em private} valuations.  Next, depending on this input, the items are allocated and every bidder is charged a price. An interesting aspect of this problem is that  a  bidder  is driven by her own selfish interest. As a result, she might misreport her valuation while trying to  manipulate the scheme to her advantage. An auction is  {\em truthful} if it is robust to such manipulations.  Our goal is to find a truthful auction that maximizes the seller's {\em revenue} - the sum of the payments made by the bidders.

It is easy to see that  there is no truthful auction that gives good revenue on every input. This motivates the   {\em Bayesian approach} to Auction Theory: Assume that the private valuation of each bidder is drawn from a publicly known prior distribution, and design a truthful auction that maximizes the {\em expected} revenue. Here, the expectation is over the priors of all the bidders, and the random choices made by the auctioneer.

For example, if we have one item, and one bidder whose valuation is  a random variable $v$,  then the optimal  auction offers the item to the bidder at some price $p$ that maximizes the expected revenue $p \times \pr[v \geq p]$.  This revenue-maximizing price is also known as the bidder's {\em reserve-price}. 

The Bayesian approach can be justified from two different standpoints. (1) A prior  can  be interpreted as a third person's uncertain {\em belief} over the private type of the bidder. (2) Alternately, we can think of a prior to be the {\em average case} scenario,  when the same scheme is executed in multiple occasions and the bidder's type might change over time. This latter framework has an interesting implication, as described below.

Ideally, our auction should not be sensitive to these priors. Instead, the  auction's revenue should be competitive against some appropriate {\em benchmark}, instance by instance. To complete the picture,  we should  guarantee that in expectation over the priors, this benchmark is close to the optimal Bayesian revenue. This will imply that  our auction is {\em prior-free}, and  still performs well on the average.

This paper deals with prior-free auction for {\em digital goods}. In this setting, the bidders are unit-demand, the items are homogeneous and indivisible, and  available in unlimited supply. 

First, consider  the scenario when the  bidders' valuations for the items are drawn from i.i.d. distributions. Here, the optimal Bayesian auction sets the same reserve price (say $q$) for every bidder: If a bidder's valuation is at least $q$, then she takes one copy of the item and makes a payment of $q$, otherwise she receives  zero items and makes zero payment. This leads to  the following prior-free revenue-benchmark. 

Let $\mathcal{F}$ denote the maximum revenue one can get by setting a uniform price across all the bidders. Now, if our revenue  is $O(1)$-competitive against $\mathcal{F}$, instance by instance,  then it is easy to see that (in expectation over the priors) we shall  be $O(1)$-competitive against the optimal Bayesian revenue. Unfortunately, for a technical reason, no truthful auction can compete against the benchmark $\mathcal{F}$ on every input. To see this, suppose that one bidder has infinite valuation, and every other bidder has zero valuation. As a result, the benchmark $\mathcal{F}$ is set at infinity. In contrast, in any truthful auction, the price offered to the first bidder should be independent of her own valuation. Hence, there is no truthful auction whose revenue is comparable to $\mathcal{F}$.

To circumvent this difficulty, Goldberg et al.~\cite{Goldberg-Games2006} proposed the modified benchmark $\F$: It is the maximum revenue one can get by setting a uniform price, under the condition  that this price is at most the second highest valuation (so that at least two bidders purchase the item). This rules out the bad example described in the preceding paragraph. In their seminal work, Goldberg et al.~\cite{Goldberg-Games2006} gave a truthful prior-free auction whose revenue is constant competitive against $\F$ on every input.

More recently, Leonardi et al.~\cite{Tim-Leonardi-stoc2012} extended the landscape of this problem by relaxing the iid assumption. They allowed the bidders' valuations to be drawn from independent, but not necessarily identical, prior distributions: As in the iid case, here the optimal Bayesian auction offers one copy of the item to every bidder at her reserve price. However, in contrast to the iid case,  different bidders may have different reserve prices. Thus, at first glance, the natural prior-free benchmark seems to be the maximum revenue  that can be obtained from any arbitrary pricing. Unfortunately, it is easy to see that this benchmark is equal to the social welfare, and hence, no truthful auction can be competitive against this benchmark. 

 Leonardi et al.~\cite{Tim-Leonardi-stoc2012} argued that in order to get a positive result in the non-iid setting, the auctioneer should be equipped with some information regarding the bidders' priors. They made the following natural assumption: The auctioneer  knows the ordering of the bidders' reserve prices.\footnote{We emphasize that the auctioneer does not know the actual values of these reserves.} Without any loss of generality,  the bidders are ordered according to $1, \ldots, n$, so that for all $i < j$, the reserve price for bidder $i$ is at most the reserve price for bidder $j$. 

In this scenario, the optimal Bayesian revenue is attained by an  increasing price ladder: For any two bidders $i < j$, the price offered to bidder $i$ is at most the price offered to bidder $j$. Due to a technical reason similar to the one described for iid bidders, one needs to impose the further condition that the highest price in this ladder is at most the second highest valuation. The maximum revenue from any such price ladder is termed as the {\em monotone-price benchmark} $\M$. Leonardi et al. gave a truthful auction that is $O(\log^* n)$ competitive against $\M$.

\paragraph{Our Result} {\em We improve upon the work of Leonardi et al.~\cite{Tim-Leonardi-stoc2012}, by giving the first truthful prior-free auction that is $O(1)$ competitive  against the monotone-price benchmark $\M$.}

\paragraph{Remark} The same result has been obtained independently in~\cite{arxiv}.

\subsection{Previous Work}

The Bayesian approach to auction theory was pioneered by Myerson~\cite{myerson}. He described the revenue-optimal auction for selling one item. For auctions involving multiple items,  his result holds  if the bidders' valuation functions (defined over subsets of items) are {\em single-dimensional}.  Several recent papers~\cite{shuchi,shuchi1,sayan,costis1,costis2}  extend  Myerson's work, and algorithmically characterize the revenue-optimal auctions in  multi-dimensional settings.

As mentioned before, Goldberg et al.~\cite{Goldberg-Games2006} were the first to introduce the concept of prior-free auctions for digital goods (see the survey in~\cite{Hartline-Karlin-bookchapt}), and to define the uniform-price benchmark $\F$. Based on a {\em random-paritioning} scheme, they presented a truthful auction that is constant-competitive against $\F$. In this scheme, the auctioneer randomly selects a ``training'' set of bidders, and uses the bids reported by this training set  to come up with an appropriate price-vector, which in turn, specifies the prices at which the item is offered   to the remaining bidders. This work was followed up by a series of papers~\cite{alei,ichiba,feige}. At present, the auction of Ichiba et al.~\cite{ichiba} has the best known competitive ratio of $3.12$ against the benchmark $\F$. 

The monotone-price benchmark $\M$ considered in~\cite{Tim-Leonardi-stoc2012} follows from the work of Hartline et al.~\cite{Tim-Hartline-stoc2008}, who  laid the foundation for deriving suitable prior-free revenue-benchmarks for a wide variety of settings.

Another research direction that has been studied extensively   is the topic of  {\em prior-independence}. This line of work is similar to the prior-free auctions, in the sense that the auctioneer is not aware of  the distributions from which the bidders' private valuations are drawn. However, in contrast to the prior-free setting, here the auction does not compete against some worst-case revenue-benchmark. Instead,  the prior distributions are used to directly analyze the auction's expected revenue. We refer the reader to the papers~\cite{Devanur-WINE11-PriorIndependent, Tim-PriorIndenpent-EC11} for more details.

\section{Notations and Preliminaries}

There are a set $\I = \{1, \ldots, n\}$ of  unit-demand bidders, and an indivisible item which is available in unlimited supply.  Bidder $i \in \I$ has a private valuation $v_i$ for the item.  The notation $\V= (v_1, \ldots, v_n)$  represents the valuation-profile of the bidders.

In an auction, the bidders  first disclose their  valuations. We use the symbol $b_i$  for the reported valuation of bidder $i \in \I$, which is called her ``bid''. Depending on the input bid-vector  $ \mathbf{b} = (b_1, \ldots, b_n)$, the auctioneer  allocates the items and every bidder is charged a  price. 

Let $X_i(\mathbf{b}) \in \{0,1\}$ be the indicator variable corresponding to the allocation of bidder $i \in \I$, and let $p_i(\mathbf{b})$ denote her payment. The auctioneer's revenue is equal to $\sum_{i \in \I} p_i(\mathbf{b})$.

The utility of bidder $i \in \I$   is $u_i(\mathbf{b}) = v_i \cdot X_i(\mathbf{b}) - p_i(\mathbf{b}).$ 
The notation $\mathbf{b}_{-i}$ represents the profile of all bids  except that of bidder  $i \in \I$. In a truthful auction, we have:
$$ u_i(\mathbf{b}_{-i},v_i) \geq \max \left\{0, u_i(\mathbf{b}_{-i},b_i)\right\}  \ \mbox{ for all } i, \mathbf{b}_{-i}, b_i, v_i.$$
In other words, an auction is truthful iff no bidder can gain in utility by misreporting her  valuation.  

A randomized auction is a probability distribution over deterministic auctions. It is truthful iff all the deterministic auctions in its support are truthful. 

In this paper, we  only consider truthful auctions. Hence, without any loss of generality, every bidder's reported bid coincides with her private valuation, and we  use the notations $\mathbf{b}$ and $\mathbf{v}$ interchangeably.

For any subset of bidders $S \subseteq \I$,  let $b^{(2)}_S$ denote the second highest bid in the set $\{ b_i : i \in S\}$.  A {\em price vector} $\p_{S}$   offers the item at price $\mathbf{P}_S(i)$ to every bidder $i \in S$. 

\begin{define}
\label{def:uniform} 
A price vector $\mathbf{P}_S$  over $S \subseteq \I$ is {\em uniform} iff
\begin{itemize} 
\item  For all  $i, j \in S$, we have $\p_{S}(i) = \p_S(j) = q$ (say).  
\item Furthermore, we have $q \leq b^{(2)}_{S}$.
\end{itemize}
\end{define}

\begin{define}
\label{def:monotone}
A price vector $\p_S$  over $S \subseteq \I$  is {\em monotone} iff
\begin{itemize}
\item We have $\mathbf{P}_S(i) \leq  \mathbf{P}_S(j)$ whenever $i < j$ and $i, j \in S$.
\item Furthermore, we have $\max_{i \in S} \left\{ \mathbf{P}_S(i) \right\} \leq b^{(2)}_S.$
\end{itemize}
\end{define}

Note that the concept of a monotone price vector  implies an underlying ordering of the set of  bidders $\I = \{1, \ldots, n\}$. For ease of notation, we  assume that bidder $i \in \I$ comes before bidder $j \in \I$ in this ordering iff $i < j$. The prices  must respect this ordering, meaning that if bidder $i$ is offered a price, then it should be  lower than  the price offered to every bidder $j > i$. In contrast, the valuations of the bidders can be arbitrary, so that we might have $b_i > b_j$ for two bidders $j > i$.


In a price vector $\mathbf{P}_S$ defined over $S \subseteq \I$, a bidder $i \in S$ takes the item iff she gets nonnegative utility (i.e., iff $b_i \geq \mathbf{P}_S(i)$). Thus, the revenue from $\p_{S}$ is given by:
$$\text{{\sc Rev}}(\mathbf{P}_S) = \sum_{i \in S \, : \, b_i \geq \mathbf{P}_S(i)} \mathbf{P}_S(i).$$

We are now ready to define the  prior-free revenue-benchmarks.

\begin{define}
\label{def:benchmarks}
 The {\em uniform-price benchmark} $\F_S$ is the maximum revenue  from a  uniform price vector defined over $S \subseteq \I$.
$$\F_S = \max_{\p_S \text{ is uniform}} \rv(\p_S).$$
\end{define}

\begin{define}
\label{def:benchmark:monotone}
The {\em monotone-price benchmark} $\M_S$ is  the maximum revenue  from a  monotone price vector  defined over $S \subseteq \I$.
$$\M_S = \max_{\p_S \text{ is monotone}} \rv(\p_S).$$
\end{define}

We  present a truthful auction whose revenue is at least a constant fraction of the monotone-price benchmark $\M_{\I}$. We do not  try to optimize the value of this constant, which is left as an interesting open question.



\section{The Auction}

Suppose that we have partitioned the bidders into two subsets $S \subseteq \I$ and $T = \I \setminus S$. Consider a price vector $\p_S$ (resp. $\p_T$) defined over the subset $S$ (resp. $T$).  We say that $\p_T$ is the {\em extension} of $\p_S$ if and only if for all $i \in T$, we have: 
\begin{eqnarray}
\p_T(i) = \begin{cases} \max_{j \in S \, : \, j < i} \p_ S(j) & \text{if } \{ j \in S : j < i \} \neq \emptyset ;\\
0 & \text{otherwise.}
\end{cases} 
\end{eqnarray} 
Let $Q_{2}$ denote the set of prices that are in powers of $2$, that is, $Q_{2} = \{ 2^{t} : t \in \mathbb{Z} \}$, where $\mathbb{Z}$ is the set of all integers. We say that a price vector $\p_S$ is {\em discretized} if and only if $\p_S(i) \in Q_{2}$ for all agents $i \in S$.

Our auction  is based on the  random partitioning framework, and it is described in Figure~\ref{fig:auction}.

\begin{figure}[htbp]
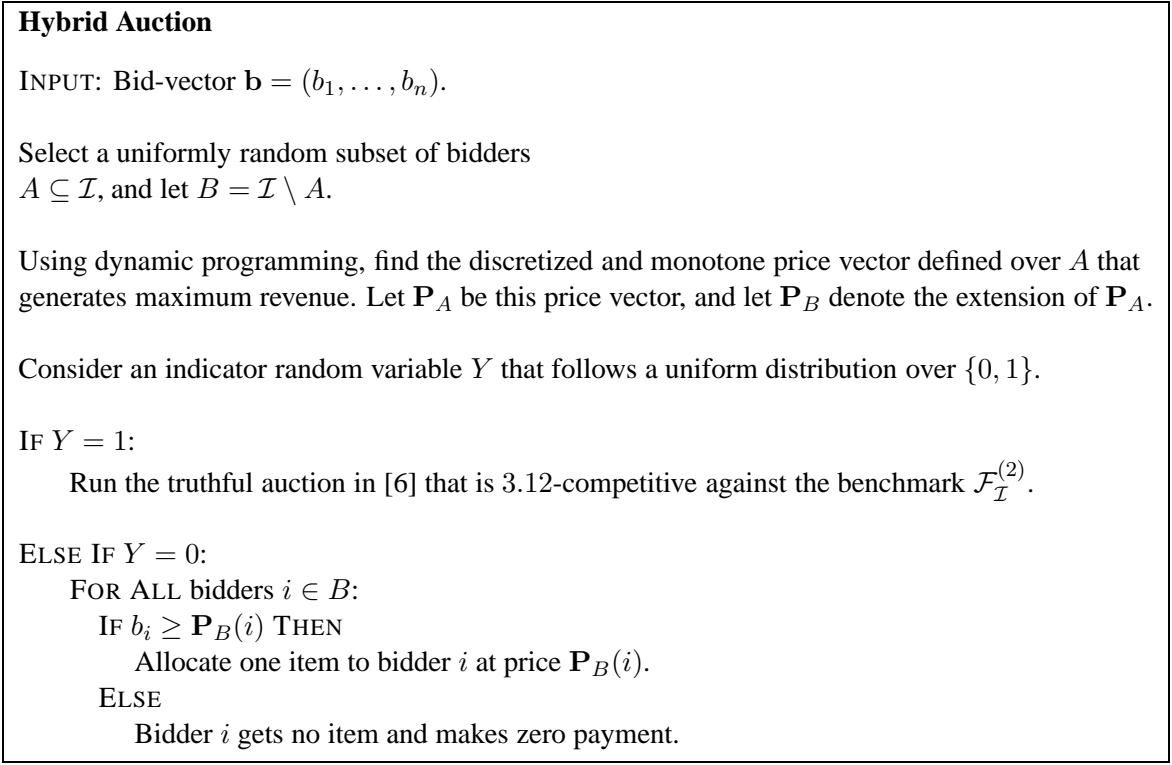

\centerline{\framebox{
\begin{minipage}{5in}
{\bf Hybrid Auction}
\begin{tabbing}
\= {\sc Input:} Bid-vector $\mathbf{b} = (b_1, \ldots, b_n)$.  \\ \\
\>     Select a uniformly random subset of bidders  \\
\>  $A \subseteq \I$, and  let $B = \I \setminus A$.  \\ \\
\>  Using dynamic programming, find the discretized    and monotone price vector defined over $A$ that \\
\>   generates maximum revenue. Let $\mathbf{P}_A$ be this price    vector, and let $\mathbf{P}_B$ denote the extension of $\mathbf{P}_A$. \\ \\
\> Consider an indicator random variable $Y$ that   follows a uniform distribution over $\{0,1\}$. \\ \\   
\> {\sc If}  $Y = 1$: \\ 
\>  \ \ \ \ \ \ \ \= Run the truthful auction in~\cite{ichiba} that is   $3.12$-competitive against the benchmark $\F_{\I}$. \\ \\
\>  {\sc Else If} $Y = 0$:  \\ 
\> \>   {\sc For All}  bidders $i \in B$: \\
\> \>  \ \ \ \ \= {\sc If}  $b_i \geq \mathbf{P}_B(i)$ {\sc Then} \\
\>  \>  \> \ \ \ \ \ \= Allocate one  item to bidder $i$ at price $\mathbf{P}_B(i)$. \\
\>  \> \> {\sc Else} \\
\>  \> \> \> Bidder $i$  gets no item and makes zero payment.
  \end{tabbing}
\end{minipage}
}}
\caption{\label{fig:auction} A truthful auction that is constant competitive against the monotone price benchmark $\M_{\I}$.} 
\end{figure}

 With probability $1/2$, we run a truthful auction whose revenue is within a constant fraction of the uniform-price benchmark $\F_{\I}$. With the remaining probability, we run the general scheme: Include every bidder $i \in \I$ in subset $A$ with probability $1/2$, independently of the other bidders. Next,  find the revenue maximizing discretized monotone price vector $\p_A$, and apply the extension of $\p_A$ to the bidders in $B$.

\begin{thm}
The Hybrid Auction (see Figure~\ref{fig:auction}) is truthful.
\end{thm}

\begin{proof}
We only need to show that the auction is truthful when we run the general scheme, i.e. when $Y = 0$. The random subset $A \subseteq \I$ is chosen independently of the input bid-vector $\B$. If a bidder $i \in \I$ is included in the subset $A$, then she gets zero utility. On the other hand, if she is included in the subset $B$, then she is offered the item at a price that is independent of her own bid. The theorem follows. 
    \end{proof}

\subsection{Revenue Guarantee}

To show that our auction is $O(1)$-competitive against the benchmark $\M_{\I}$, we shall  consider two  cases. 

Case 1.  The ratio $\F_{\I}/\M_{\I}$ is at least some constant. Note that with  probability $1/2$, we execute an auction  which is $O(1)$-competitive against $\F_{\I}$. Hence, in this case, our  revenue is clearly within a constant factor of $\M_{\I}$.

Case 2. The ratio $\F_{\I}/\M_{\I}$ is very small. If this is the case, then we  prove that   the expected value of  $\rv(\p_B)$ is within a constant factor of $\M_{\I}$. Note that with   probability $1/2$, we run the general scheme whose revenue is given by $\rv(\p_B)$. Hence, we remain $O(1)$-competitive against $\M_{\I}$.

To carry out the plan described above, we need to introduce some definitions.

\begin{define}
\label{def:label}
For any price $q \in Q_{2}$ and any  integer $l \geq 0$, we say that $q$ is a level-$l$-price iff 
$$\M_{\I}/2^{l+1} < q \leq \M_{\I}/2^l.$$
\end{define}

Since the upper and lower limits of a level differ by a factor of $2$,  there is exactly one level-$l$-price in $Q_{2}$. Throughout  this paper, we shall use the symbol $q_l$ to denote this unique level-$l$-price in $Q_2$.

\begin{define}
\label{def:triple}
Consider an ordered pair of bidders $i < j$, and the level-$l$-price $q_l \in Q_{2}$. If both the bidders'   valuations are larger than the price (that is, $b_i \geq q_l$ and $b_j \geq q_l$), then we say that $(i,j,q_l)$ is a level-$l$-triple.  
\end{define}

The concept of a triple uses the underlying ordering of the set  $\I$. Accordingly, a bidder $k \in \I$ belongs to the triple $(i,j,q_l)$ iff $i \leq k \leq j$.  A bidder in a triple $(i,j,q_l)$ is {\em winning} iff her valuation is larger than the price $q_l$. 

\begin{define}
\label{def:win}
The set of winning bidders in a triple $(i,j,q_l)$ is defined as: $$W_{ijq_l} = \{ k \in \I : i \leq k \leq j  \text{ and } b_k \geq q_l\}.$$ 
\end{define}

We say that a triple  is {\em balanced} iff its winning bidders  are evenly partitioned among the random subsets $A, B$.
\begin{define}
\label{def:balance}
A triple $(i,j,q_l)$ is balanced iff $$\frac{1}{3}\times |W_{ijq_l}| \leq |A \cap W_{ijq_l}| \leq \frac{2}{3} \times |W_{ijq_l}|.$$
\end{define}

Finally, we say that a triple is {\em large} if it contains sufficiently many winning bidders.

\begin{define}
\label{def:large}
A level-$l$-triple $(i,j,q_l)$ is large iff   $$|W_{ijq_l}| \geq 288 l.$$
\end{define}

The rest of the paper is organized as follows. In Section~\ref{subsec:events}, we show that certain important events occur with constant probability. In Section~\ref{subsec:analysis}, we show that conditioned on these important events, our auction  generates good revenue.

\subsubsection{Important Events}
\label{subsec:events}

Let $\e_1(\B)$ denote the event that $\text{{\sc Rev}}(\p_A) \geq \M_{\I}/6$.   Let $\e_2(l,\B)$ denote the event that every large level-$l$-triple is balanced. Define the  event $\e_2(\B)$ as follows. 
\begin{equation}
\label{eq:secondevent}
\e_2(\B) = \bigcap_{l \geq 24} \e_2(l,\B)
\end{equation}

We shall show that both the events $\e_1(\B)$ and $\e_2(\B)$ occur simultaneously with probability at least $1/32$ (see Theorem~\ref{cor:main}).

\begin{lemma}
\label{thm:good:1}
We have: $\pr\left[\e_1(\B)\right]  \geq \frac{1}{16}.$
\end{lemma}

\begin{proof}
Leonardi et al. proved that $\M_{\A} \geq \M_{\I}/3$ with probability at least $1/16$ (see Lemma~3.2 in~\cite{Tim-Leonardi-stoc2012}). Since $\p_A$ is the discretized monotone price vector with maximum revenue, we have $\rv(\p_A) \geq \M_{\A}/2$. The lemma follows.
 \end{proof}

\begin{claim}
\label{lm:triple}
For every integer $l \geq 0$, the number of level-$l$-triples is at most $2^{2l+2}$.
\end{claim}

\begin{proof}
  Consider a bidder $k$ whose valuation $b_k$ is at least $q_l$. Since  $q_l > \M_{\I}/2^{l+1}$, we infer that $b_k > \M_{\I}/2^{l+1}$. Thus, there are at most $2^{l+1}$ such bidders. Since a level-$l$-triple $(i,j,q_l)$ is uniquely determined by two bidders $i < j$ having valuations at least $q_l$,  it is easy to see that there can be at most $(2^{l+1})^2 = 2^{2l+2}$ level-$l$-triples. 
    \end{proof}

We shall use the following version of the Chernoff bound~\cite{Motvani-Raghavan}.

\begin{theorem}
\label{thm:chernoff}
Let $T_1, \ldots, T_m$ be i.i.d random variables such that $T_i \in \{0,1\}$ for all $i \in \{1, \ldots, m\}$. Define their sum as $T = \sum_{i=1}^m T_i$, and let $\mu = E[T]$. For all $0 < \delta < 1$:
$$\pr[(1-\delta) \mu \leq T \leq (1+\delta) \mu] \geq 1 - 2 \times \exp\left(-\frac{\mu \delta^2}{4}\right).$$
\end{theorem}

\begin{claim}
\label{lm:triple:prob}
We have: $$\pr[\e_2(l,\B)] \geq 1 - 1/2^{l}, \ \  \text{ for all } l \geq 24.$$
\end{claim}

\begin{proof}
Fix any large level-$l$-triple $(i,j,q_l)$. By definition, the number of winning bidders in $(i,j,q_l)$ is at least $288l$. Since each of these bidders is included in the set $A$ independently and uniformly at random,  Theorem~\ref{thm:chernoff} implies that  the triple $(i,j,q_l)$ is not balanced with probability at most $2/e^{4l}$. By Claim~\ref{lm:triple}, there are at most $2^{2l+2}$ level-$l$-triples. Applying union bound,  the probability that some level-$l$-triple is not balanced is at most $2^{2l+2} \times 2 /e^{4l} \leq 1/2^l$,   for  $l \geq 24$. 
    \end{proof}

\begin{lemma}
\label{thm:good:2}
We have: $\pr[\e_2(\B)] \geq 31/32$.
\end{lemma}

\begin{proof}
Applying union-bound, we infer that
$$1 - \pr[\e_2(\B)]  \leq  \sum_{l \geq 24} (1 - \pr[\e_2(l,\B)]) 
 \leq  \sum_{l \geq 24} \frac{1}{2^l} \leq \frac{1}{32}.$$
    \end{proof}

\begin{theorem}
\label{cor:main}
We have: $\pr[\e_1(\B) \cap \e_2(\B)] \geq 1/32$.
\end{theorem}

\begin{proof}
Follows from applying union bound on Lemma~\ref{thm:good:1} and Lemma~\ref{thm:good:2}.
    \end{proof}

\subsubsection{Main Analysis}
\label{subsec:analysis}

Under  the monotone price vector $\mathbf{P}_A$, let $A_l \subseteq A$ denote the subset of  bidders   who take the item at the level-$l$-price $q_l$. 
$$A_l = \{ i \in A : b_i \geq \mathbf{P}_A(i), \text{ and } \mathbf{P}_A(i) = q_l\}.$$
Note that  $A_l$ respects the  ordering of the set of  bidders $\I = \{1, \ldots, n\}$. Specifically, if we have three bidders $i < j < k$  such that  (a) $i,j,k \in A$, (b) each of them has a valuation at least $q_l$, and (c)  $i,  k \in A_l$; then we must have that $j \in A_l$.

Let  $\text{{\sc Rev}}(\p_A,l)$ denote the contribution towards   $\rv(\p_A)$  by the bidders in $A_l$. Accordingly, we get:  
$$\text{{\sc Rev}}(\p_A,l) = |A_l| \times q_l, \text { and } \rv(\p_A) = \sum_{l \geq 0} \text{{\sc Rev}}(\p_A,l).$$

\begin{define}
\label{def:goodlevel}
A level $l$ is {\em good} w.r.t. $\mathbf{P}_A$ iff $l \geq 24$ and $|A_l| \geq 288l$; otherwise the level  is {\em bad}. 
\end{define}

 The next claim implies that  all the bad levels $l \geq 24$ contribute relatively little towards $\text{{\sc Rev}}(\p_A)$.

 \begin{claim}
 \label{cl:badlevel}
We have:
 $$\sum_{l \geq 24 \ : \ l \text{ is bad }} \text{{\sc Rev}}(\p_A,l) \leq \frac{1}{18}\times \M_{\I}.$$
   \end{claim}

 \begin{proof}
 Consider any level $l \geq 24$ that is bad. Since $|A_l| < 288l$ and the level-$l$-price in $\mathbf{P}_A$ is at most $\M_{\I}/2^{l}$, we have:
 $$\rv(\p_A,l) = |A_l| \times q_l <  \frac{288l}{2^l} \times \M_{\I}.$$ Summing over all such levels, we get
\begin{eqnarray*}
\sum_{l \geq 24 \, : \, l \text{ is bad}} \rv(\p_A,l) & \leq & \sum_{l \geq 24} \frac{288l}{2^l} \times \M_{\I} \\
& \leq &  \sum_{l \geq 24} \frac{1}{2^{l/2}} \times \M_{\I} \\
& \leq & \frac{1}{18} \times \M_{\I}.
\end{eqnarray*}
  \end{proof}
 
Now, we are ready to prove the revenue guarantee.

\begin{theorem}
\label{thm:} The expected revenue of the Hybrid Auction (see Figure~\ref{fig:auction}) is within a constant factor of the monotone-price benchmark $\M_{\I}$.
\end{theorem}

\begin{proof}
We shall consider two  mutually exclusive and exhaustive cases, and show that in the first (resp. second) case, the expected revenue of our auction is at least $1/2700$ (resp. $1/2304$) times the benchmark $\M_{\I}$.

\paragraph{Case 1} $432 \times \F_{\I} \geq \M_{\I}$. 

With probability $1/2$, we execute the auction of Ichiba et al.~\cite{ichiba}, which is  $3.12$ competitive against the benchmark $\F_{\I}$. Hence, the expected revenue of our auction scheme is at least $\F_{\I}/(2\times 3.12)$, which in turn, is lower bounded by $\M_{\I}/2700$.

\paragraph{Case 2} $432 \times \F_{\I} < \M_{\I}$. 

We first claim:
\begin{equation}
\label{eq:firstlevels}
\sum_{l=0}^{23} \rv(\p_A,l) \leq \M_{\I}/18
\end{equation}
For the sake of contradiction, suppose that the above equation is not true. In that case, there is some level $l^* \in [0,23]$ with $\rv(\p_A,l^*) = |A_{l^*}| \times q_{l^*} > \M_{\I}/(18 \times 24)$. 

Now, consider the price vector $\mathbf{P}'_{\I}$ which offers the item at price $q_{l^*}$ to every bidder: We have $\mathbf{P}'_{\I}(i) = q_{l^*}$ for all $i \in \I$.  Since $\mathbf{P}_A$ is a monotone price vector  and $A_{l^*}$ is non-empty, there are at least two bidders in $A$ whose bids are lower bounded by  $q_{l^*}$.  Hence, the price vector $\mathbf{P}'_{\I}$ is uniform, and we get:  
$$\F_{\I} \geq \rv(\p'_{\I}) \geq \rv(\p_A,l^*) \geq \M_{\I}/(18 \times 24).$$ This contradicts our assumption that  $432 \times \F_{\I} < \M_{\I}$. Accordingly, we infer that equation~(\ref{eq:firstlevels}) must be true.

Next, note that conditioned on the event $\e_1(\B)$, we have 
\begin{equation}
\label{eq:condition}
\rv(\p_A) \geq \frac{1}{6} \M_{\I}
\end{equation}

From equations~(\ref{eq:firstlevels}),~(\ref{eq:condition}) and Claim~\ref{cl:badlevel}, it follows that
\begin{equation}
\label{eq:condition2}
\text{Under event  } \e_1(\B), \qquad \sum_{l \text{ is good}} \rv(\p_A,l) \geq \frac{1}{18} \M_{\I} 
\end{equation}

For the rest of the proof,  condition on the event $\e_1(\B) \cap \e_2(\B)$.

Consider any good level $l$. Let the first (resp. last) bidder in $A_l$ be denoted by $i$ (resp. $j$): For all $k \in A_l$, we have $i \leq k \leq j$. Since  $\mathbf{P}_A(i) = \mathbf{P}_A(j) = q_l$ and $b_i, b_j \geq q_l$, we infer that $(i,j,q_l)$ is a level-$l$-triple. The  number of winning bidders in this triple is at least $|A_l|$. Since the level $l$ is good, we have $|A_l| \geq 288l$. Thus, we conclude that the triple $(i,j,q_l)$ is large. Furthermore, since the level $l$ is good, we have $l \geq 24$.

By definition,  the triple $(i,j,q_l)$ is balanced under the event $\e_1(\B) \cap \e_2(\B)$. In other words, at least one third of its winning bidders  are assigned to the set $B$. Thus, we get:
$$|W_{ijq_l} \cap B| \geq \frac{1}{2} |W_{ijq_l} \cap A| = \frac{1}{2} |A_l|.$$
Furthermore, since the price vector $\p_B$ is the extension of $\p_A$,  we have $\mathbf{P}_B(k) = q_l$ for all bidders $k \in W_{ijq_l} \cap B$. Hence, the contribution towards $\text{{\sc Rev}}(\p_{B})$ by such bidders is at least $\text{{\sc Rev}}(\p_A,l)/2$. 
Summing over all good levels, we conclude:
$$\text{{\sc Rev}}(\p_B) \geq \frac{1}{2} \times \sum_{l \text{ is good}} \text{{\sc Rev}}(\p_A,l) \geq \frac{1}{36} \times \M_{\I}.$$
The second inequality follows from equation~(\ref{eq:condition2}).

To summarize,   the event  $\e_1(\B) \cap \e_2(\B)$ occurs with probability at least $1/32$ (see Theorem~\ref{cor:main}), and  conditioned on this event, we have $\text{{\sc Rev}}(\p_B) \geq \M_{\I}/36$. Finally,  we execute this scheme with probability $1/2$. Putting all these observations together,  the expected revenue of our auction is at least
$$\frac{1}{2} \times \frac{1}{32} \times \frac{\M_{\I}}{36} = \frac{\M_{\I}}{2304}.$$ 
    \end{proof}

\bibliographystyle{plain}



\end{document}